\documentclass{article}

\usepackage{amsmath,amssymb, amsfonts,latexsym,bbm,xspace,graphicx,float,mathtools}
\usepackage[backref,colorlinks,citecolor=blue,bookmarks=true]{hyperref}
\usepackage[letterpaper,margin=1in]{geometry}

\usepackage{ltexpprt}

\usepackage{sodafixes}
\usepackage{nopageno}
\usepackage{color}
\usepackage{comment}
\usepackage{algorithm}
\usepackage{algorithmic}

\usepackage{tweaklist}
\renewcommand{\enumhook}{\setlength{\topsep}{2pt}%
  \setlength{\itemsep}{-2pt}}
\renewcommand{\itemhook}{\setlength{\topsep}{2pt}%
  \setlength{\itemsep}{-2pt}}

\usepackage{fullpage}
\usepackage{appendix}
\usepackage[nomessages]{fp}

\newtheorem{definition}{Definition}

\newtheorem{observation}{Observation}

\newcommand{\Dim}{\text{Dim}}

\newenvironment{prevproof}[2]{\noindent {\em {Proof of {#1}~\ref{#2}:}}}{$\Box$\vskip \belowdisplayskip}

\newcommand{\junk}[1]{}

\junk{

}

\newcommand{\poly}{{\rm poly}}

\newcommand{\mattnote}[1]{{\color{blue}{#1}}}

\newcommand{\notshow}[1]{{}}

\DeclareMathOperator{\E}{E}
\definecolor{MyGray}{rgb}{0.8,0.8,0.8}

\begin{document}
\title{Interpolating Between Truthful and non-Truthful Mechanisms for Combinatorial Auctions}

\author{
Mark Braverman \thanks{Department of Computer Science, Princeton University, email: mbraverm@cs.princeton.edu. Research supported in part by an NSF CAREER award (CCF-1149888), NSF CCF-1215990, a 
Turing Centenary Fellowship, a Packard Fellowship in Science and Engineering, and the Simons Collaboration on Algorithms and Geometry.}
\and
Jieming Mao  \thanks{Department of Computer Science, Princeton University, email: jiemingm@cs.princeton.edu}
\and 
S. Matthew Weinberg \thanks{Department of Computer Science, Princeton University, email: sethmw@cs.princeton.edu }
}

\addtocounter{page}{-1}
\maketitle
\begin{abstract} We study the communication complexity of combinatorial auctions via \emph{interpolation mechanisms} that interpolate between non-truthful and truthful protocols. Specifically, an interpolation mechanism has two phases. In the first phase, the bidders participate in some non-truthful protocol \emph{whose output is itself a truthful protocol}. In the second phase, the bidders participate in the truthful protocol selected during phase one. Note that virtually all existing auctions have either a non-existent first phase (and are therefore truthful mechanisms), or a non-existent second phase (and are therefore just traditional protocols, analyzed via the Price of Anarchy/Stability).

The goal of this paper is to understand the benefits of interpolation mechanisms versus truthful mechanisms or traditional protocols, and develop the necessary tools to formally study them. Interestingly, we exhibit settings where interpolation mechanisms greatly outperform the optimal traditional and truthful protocols. Yet, we also exhibit settings where interpolation mechanisms are provably no better than truthful ones. Finally, we apply our new machinery to prove that the recent single-bid mechanism of Devanur et. al.~\cite{DevanurMSW15} (the only pre-existing interpolation mechanism in the literature) achieves the optimal price of anarchy among a wide class of protocols, a claim that simply can't be addressed by appealing just to machinery from communication complexity or the study of truthful mechanisms.

\end{abstract}
\thispagestyle{empty}

\newpage

\section{Introduction}
In a combinatorial auction, a single designer has $m$ items available for purchase to $n$ buyers. Each buyer has some private valuation function $v_i(\cdot): 2^{[m]} \rightarrow \mathbb{R}^+$ over subsets of items, and the seller aims to partition the items into $S_1\sqcup\ldots\sqcup S_n$ so as to optimize the social welfare, $\sum_i v_i(S_i)$. Much recent work addresses the design of combinatorial auctions, targeting the desiderata of optimality, simplicity (from both a design and strategic perspective), and computational tractability. For instance, the celebrated Vickrey-Clarke-Groves mechansim achieves the optimal social welfare, and is \emph{truthful} (therefore it is strategically simple: no bidder need consider any strategy except for honest behavior)~\cite{Vickrey61,Clarke71,Groves73}. However, in virtually all settings of interest, the VCG mechanism is not computationally tractable, making it unusable in practice.

Much recent work of computer scientists has targeted the design of auctions that are instead approximately optimal, but computationally tractable. One active line of work searches for truthful mechanisms~\cite{LaviS05,DobzinskiNS05,DobzinskiNS06, Dobzinski07,DughmiRY11,KrystaV13}. While these results all achieve computational tractability and strategic simplicity in the strongest possible way, the mechanisms are quite involved and therefore don't achieve design simplicity. More importantly, many of these mechanisms can only guarantee approximation ratios that are polynomial in $m$. When buyers are assumed to be submodular\footnote{A valuation function $v(\cdot)$ is submodular if $v(X) + v(Y) \geq v(X \cap Y) + v(X \cup Y)$ for all $X, Y$.} or subadditive,\footnote{A valuation function $v(\cdot)$ is subadditive if $v(X) + v(Y) \geq v(X \cap Y)$ for all $X, Y$.} the best achieve ratios just logarithmic in $m$. A central open problem is the design of computationally tractable truthful mechanisms that guarantee a constant-factor approximation when valuation functions are submodular or subadditive. 

Another exciting line of work has shown simple mechanisms that achieve a low \emph{price of anarchy} at various equilibrium concepts~\cite{BhawalkarR11,PaesLemeST12,SyrgkanisT12,SyrgkanisT13,FeldmanFGL13}. These results show, for instance, that as long as buyers are subadditive and interact at equilibrium, auctioning each item simultaneously via a first-price auction achieves half the optimal social welfare~\cite{FeldmanFGL13}. All of these auctions are computationally tractable and simple in design, and many achieve approximation ratios that are very small constants, via the price of anarchy. However, \emph{none} of the equilibria at which these results hold are known to arise naturally, and some are even known to be computationally intractable~\cite{CaiP14, DobzinskiFK15}.\footnote{Some equilibrium concepts, such as a pure Nash equilibrium in simultaneous second price auctions for submodular buyers, can be found in polynomial time~\cite{LehmannLN01,ChristodoulouKS08,DobzinskiFK15}. However, the algorithms finding them are highly centralized, and the equilibria themselves are very unnatural: each item only has only one non-zero bidder, even though bidding zero on any item is possibly a dominated strategy.} Note that even distributed regret minimization may be computationally intractable in these settings, as each buyer has exponentially many (in $m$) strategies to consider. Therefore, these mechanisms are all extremely complex from a strategic perspective, as buyers would have to reason about exponentially many different strategies in order to approach an equilibrium at which good approximation guarantees hold.

So, truthful mechanisms are strategically simple but achieve poor approximation ratios, and simple mechanisms achieve good approximation guarantees but are strategically complex. As an alternative to pursuing each direction separately, we propose taking ideas from each and introduce \emph{interpolation mechanisms}. An interpolation mechanism has two phases. In the first phase, buyers participate in some non-truthful mechanism whose output is itself a truthful mechanism. In the second phase, buyers participate in the truthful mechanism selected during phase one. In this language, all truthful mechanisms are interpolation mechanisms with a non-existent first phase, and the simple mechanisms referenced above are interpolation mechanisms with a non-existent second phase. 

What might interpolation mechanisms bring to the table that truthful mechanisms and existing simple mechanisms don't? This question is best addressed with an example. Recent work of Devanur et. al.~\cite{DevanurMSW15} designs the first interpolation mechanism (although they did not consider this classification), the single-bid mechanism. Phase one of the single-bid mechanism asks each buyer to report just a single real number, $b_i$, as their bid. Phase two visits the buyers one by one in decreasing order of $b_i$, and allows the buyer to purchase any number of remaining items at $b_i$ per item (so more items are available to higher bidders, but lower bidders pay less per item). It is easy to see that once the bids are fixed and order determined in phase one, phase two constitutes a truthful mechanism. Note that phase one by itself is extremely limited: buyers are asked to represent their entire valuation function (of which there are doubly-exponentially many) with just $\log m$ bits. Unsurprisingly, no protocol using this limited amount of communication can possibly find a good allocation directly. Note also that phase two by itself is also quite limited: an ordering of the bidders along with a single price is set ahead of time, then buyers do as they please. Also unsurprisingly, such truthful mechanisms (that we call \emph{single-price} mechanisms) can't guarantee any non-trivial approximation ratio. From our perspective, the single-bid mechanism is interesting because it takes two useless mechanisms, neither of which can guarantee a sub-polynomial approximation ratio on even $0/1$-additive buyers,\footnote{A buyer is additive if they have a value $v_i$ for item $i$, and their value for a set $S$ is $\sum_i v_i$. A buyer is $0/1$-additive if each $v_i \in \{0,1\}$.} and combines them into a mechanism with a price of anarchy $O(\log m)$ at correlated equilibria when buyers are subadditive. Importantly, because the per-bidder communication in phase one is only logarithmic, each bidder can actually implement any standard regret minimization algorithm over possible bids in poly time. Therefore, the mechanism achieves design simplicity, strategic simplicity, and computational tractability. The main open problem left following their work is the design of mechanisms that achieve these three desiderata with a constant price of anarchy.

Interpolation mechanisms are a natural avenue to tackle this problem, and therefore lower bounds on their capability are important to guide their research. Following Devanur et. al.'s work, questions arose such as: what if bidders make a constant number of bids instead of just one? What if the posted prices are different for each item? What if we restrict attention to a much smaller class than subadditive bidders? What if we consider price of stability instead of anarchy? Surprisingly, a subset of our results shows that \emph{none} of these relaxations suffice to (significantly) beat the $O(\log m)$ bound attained by the single-bid mechanism. The remainder of our results show that lower bounds known for various classes of truthful mechanisms also extend to interpolation mechanisms with little first-phase communication. One should interpret these results \emph{not} as claiming that the limits of interpolation mechanisms have already been reached, but as guiding future research towards other classes of truthful mechanisms (specifically, we identify posted-price mechanisms as a natural candidate in Section~\ref{sec:discussion}). 

\subsection{Our Results}\label{sec:results}
In addition to formally identifying interpolation mechanisms as an important avenue of study, we identify their connection to the price of anarchy and stability, and provide numerous lower bounds. Our lower bounds consider interpolation mechanisms where the phase-two mechanism comes from a certain class. The goal of these lower bounds is to identify which classes of mechanisms are incompatible with interpolation (MIR and value-query, below), and for which classes of mechanisms the limits have already been reached (single-price and non-adaptive posted-price, below) to guide future research towards others (adaptive posted-price mechanisms, Section~\ref{sec:discussion}). Our results all lower bound the amount of first-phase communication necessary to find a suitable phase-two mechanism from the desired class. Note that even for truthful mechanisms, no unconditional communication lower bounds are known outside of artificial settings, so it is outside the scope of this paper to suddenly provide unconditional lower bounds in the strictly more general setting of interpolation mechanisms.

\paragraph{Price of Anarchy and Price of Stability.} Any bidder participating in an interpolation mechanism with $O(\log m)$ first-phase communication per bidder can run any standard regret minimization algorithm in poly-time. Because bidders need not strategize over their phase-two behavior, they need only optimize over their possible strategies in phase one, of which there are at most $\poly(m)$. Therefore, price of anarchy bounds for correlated equilibria of interpolation mechanisms with logarithmic first-phase communication have some extra bite, as bidders can be reasonably expected to converge to a correlated equilibria and the bound will hold. We call such mechanisms \emph{a priori learnable}, and formally define this in Section~\ref{sec:prelim}.

Because we lower bound the first-phase communication complexity, we not only lower bound the achievable price of anarchy by such mechanisms, but also the price of stability. In this context, our bounds are strong in the sense that they don't rely on equilibrium behavior of the buyers, and apply no matter how the buyers interact. Prior to this work, Roughgarden provides the only general approach for proving price of anarchy lower bounds~\cite{Roughgarden14}, and no general approach was known for price of stability at all. Our approach is similar to Roughgarden's in the sense that both identify settings in which communication lower bounds imply ``the right'' price of anarchy lower bounds. Still, our approach differs signifcantly as Roughgarden's work specifically targets equilibrium concepts that are \emph{not} efficiently computable, and doesn't apply to price of stability. We discuss formally the connection between first-phase communication bounds and price of anarchy/stability in Section~\ref{sec:prelim}. 

\paragraph{Single-Price Mechanisms.} A single-price mechanism fixes a price $p_i$ for buyer $i$, then visits the buyers one at a time and offers buyer $i$ any remaining items for $p_i$ each. Devanur et. al.'s single-bid mechanism has $O(\log m)$ first-phase communication per bidder, and obtains a price of anarchy at correlated equilibria of $O(\log m)$ whenever buyers are subadditive. We show in Section~\ref{sec:singleprice} that even when buyers are just additive, \emph{no amount of first-phase communication} suffices for an interpolation mechanism whose second phase is a single-price mechanism to obtain an approximation ratio $o(\log m/ \log \log m)$. Note that this \emph{significantly} improves a lower bound shown in~\cite{DevanurMSW15}, which simply proved that the single-bid mechanism itself could not guarantee an approximation ratio $o(\log m / \log \log m)$. 

\paragraph{Non-Adaptive Posted-Price Mechanisms.} Non-adaptive posted-price mechanisms generalize single-price mechanisms by allowing the mechanism to set a price $p_{ij}$ for buyer $i$ to purchase item $j$. The mechanism still visits the buyers one at a time, and allows buyer $i$ to purchase any remaining items at the designated price. We show in Section~\ref{sec:manyprices} that even when buyers are just additive, any interpolation mechanism whose second phase is a non-adaptive posted-price mechanism and guarantees an $o(\log m / \log \log m)$ approximation ratio has $\Omega(m^{1-\epsilon})$ first-phase communication per bidder, for all $\epsilon > 0 $. Therefore, the single-bid mechanism cannot be improved by restricting attention to a smaller class of valuations, restricting attention to a smaller class of equilibrium concepts, setting different prices for different items, or allowing significantly more (but still sublinear) first-phase communication. 

\paragraph{Maximal-In-Range, Value Query, and Computationally Efficient Mechanisms.} Several recent works have identified lower bounds on approximation ratios that can possibly be obtained by these classes of mechanisms, which we will define in the corresponding sections. We extend these lower bounds to mechanisms with low first-phase communication that induce a mechanism in one of these classes. In Section~\ref{sec:MIR}, we extend techniques of Daniely et. al. based on generalizations of the VC-dimension~\cite{DanielySS15}, and in Section~\ref{sec:value}, we extend the techniques of Dobzinski and Vondrak based on structured sub-menus~\cite{Dobzinski11,DobzinskiV12}.

\subsection{Discussion and Future Work}\label{sec:discussion} Motivated by impossibility results associated with truthful mechanisms, and concerns regarding the strategic simplicity of existing simple mechanisms analyzed via price of anarchy, we propose the study of interpolation mechanisms. Using this new notion, we show that the single-bid mechanism of Devanur et. al.~\cite{DevanurMSW15} is essentially optimal for its class, even subject to quite significant generalizations. We note that, prior to our work, it was unclear even how to define a class containing this mechanism, let alone prove lower bounds against mechanisms ``like this.'' We also identify several classes of truthful mechanisms that are incompatible with interpolation in the sense that low first-phase communication doesn't allow for better approximation guarantees than no first-phase communication. 

Our work identifies \emph{adaptive} posted-price mechanisms (where the mechanism may choose what prices to set based on what items have already sold) as an intriguing class of mechanisms to study with interpolation, as none of the lower bounds from this work apply. Furthermore, Dynkin's secretary algorithm~\cite{Dynkin63} immediately implies an adaptive posted-price mechanism that gets a $1/e$ approximation for additive bidders, so mild adaptations of our lower bounds for non-adaptive posted-price mechanisms are unlikely to apply. Can an interpolation mechanism with $O(\log m)$ per bidder first-phase communication and an adaptive posted-price mechanism for its second phase guarantee a constant price of anarchy?

Our results also fit into a line of work designing combinatorial auctions with low price of anarchy via \emph{valuation compression}~\cite{DuttingFP11,DuttingHS13,DuttingFP14,BabaioffLNP14}. These mechanisms restrict the allowable valuation reports from buyers to a space where the VCG mechanism is computationally tractable, even though the buyers may have much more complex valuations. In our context, these mechanisms still consist of just a first phase, and therefore rich valuation classes (like submodular, subadditive, or even just additive) cannot be compressed all the way down to a class that can be indexed with just $O(\log m)$ bits without super-constant loss. On this front, interpolation mechanisms provide a new style of two-phase valuation compression where this level of compression may be attainable.

{Additionally, many existing truthful auction formats are naturally parameterized by parameters that are assumed to be known to the designer (e.g. buyers' budgets in a clinching auction~\cite{GoelMP12,GoelMP13,GoelMP14} or buyers' interest sets in single-minded combinatorial auctions). Our framework provides a natural extension of such mechanisms to settings where these parameters are instead private. For instance, one could take any clinching auction where the budgets are assumed to be known, and add a first phase where buyers are asked to report their private budget. It would be very interesting to analyze the price of anarchy of such interpolation mechanisms, as these parameters are often not public knowledge in practice.}

Finally, while we were motivated to study interpolation mechanisms for welfare maximization in combinatorial auctions, interpolation will also be useful in any setting where unfortunate lower bounds are known for truthful mechanisms but strategic simplicity is still a concern. A natural generalization of the presented setting, which we omit due to space constraints, is a model where rounds of truthful and non-truthful interaction might be interleaved (instead of having all non-truthful interaction come before all truthful interaction).{ It would be interesting to understand the power and complexity of such mechanisms in settings beyond necessarily just combinatorial auctions.}

\section{Preliminaries} \label{sec:prelim}
In a combinatorial auction, the designer has $m$ items to allocate to $n$ buyers. Each item can be allocated to at most one buyer, and the buyers can be charged any non-negative price. Agents have a valuation function $v_i(\cdot)$ mapping subsets of items to non-negative real values. Agents are quasi-linear, meaning that their utility for receiving items $S_i$ and paying price $p_i$ is $v_i(S_i) - p_i$. The designer's goal is to select an allocation that (approximately) maximizes the \emph{welfare}, $\sum_i v_i(S_i)$. 

A mechanism is \emph{truthful} if it is in every buyer's interest to tell the truth, no matter their type. Formally, if $p_i(\vec{v})$ denotes the expected price paid by buyer $i$ when the reported types are $\vec{v}$, and $S_i(\vec{v})$ denotes the (possibly random) set that buyer $i$ receives, then we must have:

$$\mathbb{E}_{S_i \leftarrow S_i(\vec{v})}[v_i(S_i)] - p_i(\vec{v})$$
$$ \geq \mathbb{E}_{S_i \leftarrow S_i(\vec{v}_{-i};v'_i)}[v_i(S_i)] - p_i(\vec{v}_{-i};v'_i),\ \ \forall i,\vec{v}_{-i},v_i, v'_i.$$

We define various classes of mechanisms and subclasses of valuation functions within the following sections.

\subsection{Interpolation Mechanisms}
An interpolation mechanism is a communication protocol with two phases. The first phase is non-truthful, and the output is a truthful mechanism. The second phase is the truthful mechanism output in phase one, and the output is an allocation of items and prices to charge. 

\begin{definition}(Interpolation Mechanism) {Let $\mathcal{M}$ denote the space of all truthful mechanisms for a combinatorial auction setting. Note that the output space of all $M \in \mathcal{M}$ is an allocation of items and charged prices. An interpolation mechanism provides a communication protocol, $P$, that outputs a mechanism $M \in \mathcal{M}$ based on the transcript of $P$. In phase one, bidders participate in the protocol $P$. In phase two, bidders participate in the truthful mechanism output by $P$ during phase one. After phase two, the items are allocated and prices charged according to the bidders' play of the phase two mechanism. If the second phase of an interpolation mechanism always lies inside a restricted class $\mathcal{C}$ of truthful mechanisms, then we call this a ``$\mathcal{C}$ interpolation mechanism.''}
\end{definition}

Our main results provide lower bounds on the per-bidder communication necessary during the first phase in order to possibly select a good truthful mechanism for the second phase. Formally, we say that an interpolation mechanism guarantees an approximation ratio of $c$ when buyers have types in $\mathcal{V}$ if for all $i, v_i \in \mathcal{V}$, there exists a phase-one strategy for buyer $i$, $s_i(v_i)$, such that for all $\vec{v} \in \mathcal{V}^n$, if buyers use the strategies $s_i(v_i)$ during phase one, and report truthfully during phase two, the resulting allocation obtains a $1/c$-fraction (in expectation) of the optimal social welfare for $\vec{v}$. 

Note that this approximation guarantee is not tied to any particular equilibrium concept. It is strictly easier to design an interpolation mechanism that guarantees an approximation ratio of $c$ than one that has a price of anarchy/stability of $c$ (stated formally in the following section), so lower bounds on the approximation ratio imply lower bounds on attainable price of anarchy/stability.

Of specific interest are interpolation mechanisms {that have $\poly(n,m)$ total communication, and only require bidders to consider $\poly(m)$ strategies. Note that bidders must, at least a priori, consider every possible strategy during phase one (but need only consider telling the truth during phase two). So in order to guarantee that bidders have at most $\poly(m)$ strategies to consider, the first phase must be especially simple.}

\begin{definition}(a priori learnable) We say that an interpolation mechanism is \textbf{a priori learnable} if the first phase contains a single simultaneous broadcast round of communication, and the per-bidder communication is $O(\log m)$.\footnote{{Note that, for instance, a single simultaneous broadcast round of $\poly(m)$ communication per bidder results in exponentially many strategies (as in simultaneous first or second price auctions).}}
\end{definition}

\begin{observation}\label{obs:apriori} Any buyer can run any standard regret minimization algorithm (for instance, Multiplicative Weights Updates) over her strategies in an a priori learnable interpolation mechanism in time/space $\poly(m)$. Therefore, a correlated equilibrium of any a priori learnable interpolation mechanism can be found in poly-time, and correlated equilibria arise as the result of poly-time distributed regret minimization.
\end{observation}

\begin{proof}
As the second phase is a truthful mechanism, each buyer need not strategize over possible actions during the second phase. Therefore, buyers should always play their dominant strategies during phase two and need only learn over their strategies during phase one; this can only decrease their regret. As phase one is a normal form game and there are only $\poly(m)$ such strategies, each buyer can just run a standard regret minimization algorithm in time/space $\poly(m)$. If each buyer does this, their play will converge to a correlated equilibrium~\cite{FosterV97, HartM00}. 
\end{proof}

Note that price of anarchy bounds for correlated equilibria in a priori learnable interpolation mechanisms have more bite than price of anarchy bounds for solution concepts that don't arise naturally. The single-bid mechanism, for instance, is a priori learnable.

\subsection{Connection to Price of Anarchy and Stability}
The main application of our first-round communication lower bounds is on the price of anarchy or stability achievable for any a priori learnable interpolation mechanism. Price of anarchy/stability is typically defined for the social welfare, but has recently been considered also for revenue~\cite{HartlineHT14}, and is well-defined for more general objectives as well. The observation below holds for any objective, but we state it for social welfare in combinatorial auctions since that is the focus of this paper.

\begin{definition} Let $E$ denote any solution concept (i.e. Nash equilibria) for the mechanism $M$, and $\mathcal{V}$ denote any set of valuation functions. Then the price of anarchy (PoA) and Price of Stability (PoS) of $M$ with respect to $E$ when buyers have valuations in $\mathcal{V}$ are:
$$PoA = \max_{\vec{v} \in \mathcal{V}^n} \frac{\max_{S_1\sqcup\ldots\sqcup S_n} \{\sum_i v_i(S_i)\}}{\min_{\vec{s} \in E}\{\mathbb{E}_{S_1,\ldots,S_m \leftarrow M(\vec{s})}[\sum_i v_i(S_i)]\}}.$$
$$ PoS =  \max_{\vec{v} \in \mathcal{V}^n} \frac{\max_{S_1\sqcup\ldots\sqcup S_n}\{\sum_i v_i(S_i)\}}{\max_{\vec{s} \in E}\{\mathbb{E}_{S_1,\ldots,S_m \leftarrow M(\vec{s})}[\sum_i v_i(S_i)]\}}.$$
\end{definition}

\begin{observation}
If an interpolation mechanism has price of anarchy or price of stability $\alpha$ at any non-empty equilibrium concept, then that same interpolation mechanism guarantees an approximation ratio of $\alpha$. Therefore, lower bounds on the approximation ratios of interpolation mechanisms imply lower bounds on the possible price of anarchy/stability obtainable by those same mechanisms. 
\end{observation}

\begin{proof}
Just sample each strategy $s_i(v_i)$ from any equilibrium where the price of anarchy/stability holds. This strategy immediately witnesses that the interpolation mechanism guarantees an $\alpha$-approximation.
\end{proof}

\section{Single-Price Mechanisms}\label{sec:singleprice}
In this section, we consider \emph{single-price mechanisms}. A single-price mechanism visits bidders one at a time and offers the current bidder the opportunity to buy any number of remaining items at $p_i$ per item. The main result of this section is the following:

\begin{theorem}\label{thm:oneprice}
There exist profiles of additive buyers for which the best single-price mechanism achieves an $\Omega(\log m / \log \log m)$-approximation. Therefore, for all $C > 0$, no single-price interpolation mechanism with first-round communication $C$ per bidder obtains an $o(\log m/\log \log m)$-approximation on all profiles of additive buyers. This holds even when each buyer values each item at an integer between $1$ and $m$.
\end{theorem}

\begin{proof}
Consider the following example. There are $b$ buckets of items (indexed from $0$ to $b-1$), with bucket $i$ containing $c^{b-i}$ items, for some constants $b, c$ to be set later. The value of (almost) every bidder for each item in bucket $i$ is $c^i$. Each item is ``special'' for exactly one bidder, who values it instead at $c^{i+1}$. Each bidder has exactly $c^{b-i}/n$ special items in bucket $i$. It is clear that the optimal allocation in this instance is to award each bidder each of their special items, which has welfare $bc^{b+1}$. 

Now consider any single-price mechanism, with prices $p_1,\ldots, p_n$. We want to consider when bidder $i$ will get his special items in bucket $j$. Notice that bidder $i$'s special items in bucket $j$ are available to her if and only if $p_k >c^j$ for all $k < i$. Bidder $i$ will choose to purchase her special items in bucket $j$ if and only if $p_i \leq c^{j+1}$. 

So for each bucket $j$, let $i_j$ denote the first bidder for which $p_i \leq c^j$ (w.l.o.g. such a bidder exists as it is always optimal to set $p_n = 0$), and $n_j$ denote the number of bidders before $i$ whose price is at most $c^{j+1}$. Then the number of bidders who get their special items in bucket $j$ is exactly $n_j+1$. So the total number of pairs $(i,j)$ such that bidder $i$ gets her special items in bucket $j$ is exactly $b+ \sum_{j=1}^b n_j$. It's also clear that $\sum_{j=1}^b n_j \leq n$, as $p_i \in (c^j,c^{j+1}]$ for at most one $j$. So the number of pairs $(i,j)$ such that bidder $i$ gets her special items in bucket $j$ is at most $b+n$. 

Finally, observe that if the number of pairs $(i,j)$ such that bidder $i$ receives her special items in bucket $j$ is $x$, then the welfare is exactly $xc^{b+1}/n + (bn-x)c^b/n$, which achieves at most a $(\frac{x}{nb} + \frac{1}{c})$-fraction of the optimal welfare. Plugging in for $x = n+b$, this is a $1/(1/n+1/b+1/c)$-approximation. 

Setting $b = c = n$ provides an example with $m = \Theta(n^n)$ items (so $n = \Theta(\log m /\log \log m)$) for which no single-price mechanism obtains an $o(n) = o(\log m /\log\log m)$-approximation. 

%Setting $b = c = \frac{2n(n-\epsilon)}{\epsilon}$ provides an example for which no single-price mechanism obtains an $(n-\epsilon)$-approximation.
\end{proof}

Notice that the impossibility above is quite strong: no amount of communication suffices to find a good single-price mechanism (because it is possible that one simply doesn't exist). This greatly strengthens an inapproximability result of~\cite{DevanurMSW15}, which just shows that their specific procedure (the single-bid mechanism) for selecting one doesn't obtain a better approximation ratio.

\begin{corollary}
No single-price interpolation mechanism obtains a price of anarchy or price of stability $o(\log m /\log \log m)$ at any solution concept that is guaranteed to exist on all profiles of additive buyers.
\end{corollary}
\section{Non-Adaptive Pricing Mechanisms}\label{sec:manyprices}
In this section, we consider \emph{non-adaptive posted-price} mechanisms. A non-adaptive posted-price mechanism orders the bidders however it wants (possibly randomly), then selects a price vector $\vec{p}_i$ for each bidder $i$. The bidders are visited one at a time, and offered the opportunity to purchase any subset $S_i$ of remaining items for price $\sum_{j \in S_i} p_{ij}$. The main result of this section is \mattnote{below}. Our proof uses the probabilistic method, which has also been used in~\cite{DuttingK15} to prove price of anarchy lower bounds.

\begin{theorem}\label{thm:manyprices}
Any non-adaptive posted-price interpolation mechanism that guarantees an approximation ratio of $o(\log m /\log \log m)$ on all profiles of additive bidders has first-round communication at least $m^{1-\epsilon}$ per bidder, for all $\epsilon > 0$. This holds even when each buyer values each item at an integer between $1$ and $m$.
\end{theorem}

\begin{proof}
We will use the probabilistic method to define a set of profiles of additive bidders such that no non-adaptive posted-price mechanism does well on much of the set. Let each $\vec{v}_j$ (the vector of values of each bidder for item $j$) be drawn independently, and be equal to a random permutation of $(c^{k+1},c^k,\ldots,c^{k})$ with probability $1/c^k$ for each $k \in \{1,\ldots,b\}$, and $(0,\ldots,0)$ with probability $1-\sum_{k=1}^b 1/c^k$ for constants $c\geq 2, b$ to be set later. 

It is clear that the expected maximum value per item is exactly $bc$, so the expected optimal welfare is $bcm$. Consider now any non-adaptive posted-price mechanism, and restrict attention to prices for item $j$. For each $k$, let $i_k$ denote the first bidder such that $p_{i_kj} \leq c^{k}$, and $n_k$ denote the number of bidders before $i_k$ such that $p_{ij} \leq c^{k+1}$. Then the probability that this mechanism awards the item to the ``special'' bidder when the profile is a random permutation of $(c^{k+1},c^{k},\ldots,c^{k})$ is exactly $\frac{1+n_k}{n}$. Therefore, the expected welfare of this posted-price mechanism, just considering contributions from item $j$, is $\sum_{k=1}^b c(1+n_k)/n + (n-1-n_k)/n$. It is also clear that $\sum_{k=1}^b n_k \leq n$, as each $p_{ij} \in (c^k,c^{k+1}]$ for at most one $k$. So the expected welfare per item of this non-adaptive posted-price mechanism is at most $cb/n+c+b$, and the total expected welfare is at most $(cb/n+c+b)m$. 

Because the values for each item are drawn independently, the optimal welfare and the welfare of this non-adaptive posted-price mechanism is the sum of $m$ independent random variables, each in $[0,c^{b+1}]$. Therefore, we can use the Chernoff bound to bound the probability that these random variables deviate from their expectation. 

Set $b = c = n$. Then the probability that the welfare of any fixed item pricing exceeds $2(3n)m$ is at most $e^{-m/n^n}$. The probability that the optimal welfare is less than $(n^2m)/2$ is at most $e^{-m/(4n^{n-1})}$. So consider any set $P$ of at most $2^{m/n^n}$ different non-adaptive posted-price mechanisms. Taking a union bound over all mechanisms $M \in P$, we see that with non-zero probability, the welfare of $M$ is at most $6nm$ while the optimal welfare is at least $n^2m/2$. Therefore, there exists a profile of additive bidders for which no mechanism in $P$ is an $n/12$-approximation. 

If the first-round communication of each player is at most $m/n^{n+1}$, then there are only $2^{m/n^n}$ possible transcripts from the first round, and therefore only $2^{m/n^n}$ different non-adaptive posted-price mechanisms can possibly result. By the above reasoning, this implies the existence of a profile for which every possible mechanism selected (and therefore every outcome selected by the protocol) does not obtain an $n/12$-approximation. For any fixed $\epsilon$, setting $m = n^{n/\epsilon}$ yields an instance with $n = \Theta(\epsilon \log m /\log \log m)$ that requires $m^{1-\epsilon}$ first-round bits per bidder to optain an $n/12$ approximation. 
\end{proof}

Interestingly, there is always a non-adaptive posted-price mechanism that allocates the items optimally: set $p_{ij} = \max_{i' \neq i} v_{i'j}$ for all $i, j$. Each $v_{i'j}$ can be communicated with $\log m$ bits, so the entire mechanism can be found with $m\log m$ bits of communication per bidder. The theorem states that sublinear communication doesn't suffice to find a very good mechanism.

\begin{corollary}
No a priori learnable non-adaptive posted-price interpolation mechanism obtains a price of anarchy or price of stability $o(\log m / \log \log m)$ at any solution concept that is guaranteed to exist on all profiles of additive buyers.
\end{corollary}
\section{Maximal-In-Range Mechanisms}\label{sec:MIR}
In this section, we consider \emph{maximal-in-range} (MIR) mechanisms. A maximal-in-range mechanism selects some subset $\mathcal{F}' \subseteq 2^{[n]\times [m]}$ of feasible allocations, and \emph{always} selects an outcome in $\arg \max_{x \in \mathcal{F}'}\{\textsc{Welfare}(x)\}$ (where the welfare is computed with respect to the valuation profile). In other words, a maximal-in-range mechanism always optimizes welfare exactly over a restricted set of possible outcomes. We provide a mild generalization of the techniques of Daniely et. al.~\cite{DanielySS15} that apply to MIR interpolation mechanisms rather than just MIR mechanisms. With these new techniques, we show the following theorem. All proof details are in Appendix~\ref{app:MIR} for space considerations.

\begin{theorem}\label{thm:MIR} For all $\delta > 0$, the following hold:
\begin{itemize}
\item Assuming $NP \subsetneq P/poly$, any poly-time (runs in time $\poly(n, m)$) MIR interpolation mechanism that obtains an approximation ratio $m^{1/3-2\delta/3}$ whenever buyers are single-minded\footnote{A valuation function $v(\cdot)$ is single-minded if there is a special set $S$ and $v(T) = v(S)$ for all $S \subseteq T$, and $v_i(T) = 0$ otherwise.} has first-round communication at least $m^\delta$ per bidder.
\item Assuming $NP \subsetneq P/poly$, any poly-time (runs in time $\poly(n, m)$) MIR interpolation mechanism that obtains an approximation ratio $m^{1/3-\delta}/5$ whenever buyers are capped-additive\footnote{A valuation function $v(\cdot)$ is capped-additive if there is a budget $b$ such that $v(S) = \min\{b,\sum_{j \in S} v(\{j\})\}$.} has first-round communication at least $m^{1/3}$ per bidder.
\item Any poly-communication (total communication $\poly(n, m)$) MIR interpolation mechanism that obtains an approximation ratio $m^{1/3 - \delta}$ whenever buyers are submodular\footnote{A valuation function $v(\cdot)$ is submodular if $v(S \cup T) + v(S \cap T) \leq v(S) + v(T)$ for all $S, T$.} has first-round communication at least $m^{1/3}$ per bidder. 
\end{itemize}
\end{theorem}

\begin{corollary}
Assuming $NP \subsetneq P/poly$, no a priori learnable, computationally efficient MIR interpolation mechanism obtains a price of anarchy or price of stability $o(m^{1/3})$ at any solution concept that is guaranteed to exist on all profiles of single-minded buyers, capped-additive buyers, or submodular buyers.
\end{corollary}

\section{Value Query and Computationally Efficient Mechanisms}\label{sec:value}
In this section, we consider value query mechanisms and arbitrary computationally efficient mechanisms. A mechanism is a \emph{value query} mechanism if it only interacts with buyer valuations with queries of the form: ``what is your value for set $S$?'' A computationally efficient mechanism is any mechanism that terminates in polynomial time in $m, n,$ and the space it takes to describe a valuation function. Note that for single-minded and capped-additive functions, the space required is also $\poly(m, n)$, but for submodular functions the space required may be larger. We provide a mild generalization of techniques of Dobzinski and Vondrak~\cite{Dobzinski11,DobzinskiV12} that apply to interpolation mechanisms rather than just truthful mechanisms. With these new techniques, we show the following theorem. All proof details are in Appendix~\ref{app:general} for space considerations.

\begin{theorem}\label{thm:valuequeries} For all $\delta > 0$, the following hold:
\begin{itemize}
\item Any value query interpolation mechanism that makes at most $\frac{e^{m^{1/3}}}{10m^8}-1$ queries that obtains an approximation ratio $m^{1/3-\delta}/20$ whenever buyers have submodular valuations has first-round communication at least $m^\delta$ per bidder.
\item Assuming $RP \neq NP$, any computationally efficient interpolation mechanism that obtains an approximation ratio $m^{1/3-\delta}/20$ has first-round communication at least $m^\delta$ per bidder.
\end{itemize}
\end{theorem}

\begin{corollary}
Assuming $RP \neq NP$, no a priori learnable computationally efficient mechanism or a priori learnable value query mechanism that makes at most $\frac{e^{m^{1/3}}}{10m^8}-1$ queries guarantees a price of anarchy or price of stability $o(m^{1/3})$ at any solution concept that is guaranteed to exist on all profiles of submodular buyers.
\end{corollary}
\bibliographystyle{alpha}
\bibliography{bib} 
\appendix
\section{A Second Round is Necessary}\label{app:notruth}

In this section, we show that indeed something is gained by including a second round, if the goal is a priori learnability. Specifically, we show that just a low amount of first-round communication (but no truthful communication) can't beat a $\poly(m)$-approximation. One should contrast this with the single-bid mechanism of Devanur et. al. which shows that exponentially less first-round communication plus a single-price mechanism obtains an $O(\log m)$-approximation.

\begin{theorem}
For all $\epsilon > 0$, consider any (not necessarily truthful) protocol for combinatorial auctions with communication complexity $O(m^\epsilon)$ per party. Then this protocol does not guarantee an approximation ratio better than $m^{1/2 - \epsilon}$ on all instances with 0/1-additive buyers. This bound holds against non-deterministic and randomized protocols as well.
\end{theorem}

\begin{proof}
Let $n = 2m^{1/2 - \epsilon}$. Consider the following random valuation: Each item interests only one bidder chosen uniformly randomly and also independently from other items. Each bidder values items they are interested in at $1$, and all others at $0$.

For any fixed allocation, let's define $X_i$ to be the indicator variable that item $i$ is assigned to the bidder who is interested in this item. Then the total welfare of this allocation can be defined as $X = \sum_{i=1}^n X_i$. It's easy to see that $\E[X] = m / n$, and that the optimal allocation has welfare $m$. By Chernoff bound, we have 
\[
Pr[ X \geq 2 \E[X]]  \leq e^{-E[X] / 3} = e^{-m^{1/2+\epsilon} / 6}. 
\]
Since each party communicates only $O(m^{\epsilon})$ bits, the mechanism uses at most $2^{ n\cdot O(m^{\epsilon})}$ different allocations. And from the above probabilistic argument, we know that each allocation can only obtain a $m^{1/2 -\epsilon}$ with probability $e^{-m^{1/2+\epsilon}/6}$ on the random valuations described above. For all $\epsilon > 0$, when $m$ is large enough, 
\[
2^{ n \cdot O(m^{\epsilon})} \times e^{-m^{1/2+\epsilon}/6} < 1/m.
\]
We know that the protocol can't possibly get an $m^{1/2 -\epsilon}$ approximation on all additive valuations. Because we only counted the number of transcripts, this holds against non-deterministic protocols as well. 

In fact, we have also shown that the probability that any set of $2^{n\cdot O(m^{\epsilon})}$ allocations achieves welfare larger than $m^{1/2+\epsilon}$ is at most $1/m$. Clearly, no allocation can achieve welfare larger than $m$, so we have also shown that the expected welfare of any set of $2^{n \cdot O(m^{\epsilon})}$ allocations is at most $m^{1/2+\epsilon}$ on this distribution as well. Now we can apply Yao's minimax principle to claim that no randomized protocol beats welfare $m^{1/2+\epsilon}$ on all instances generated by our probabilistic construction. Therefore, no randomized protocol guarantees a $m^{1/2-\epsilon}$-approximation on all instances of 0/1-additive buyers.
\end{proof}

\section{Omitted Proofs from Section~\ref{sec:MIR}}\label{app:MIR}

In this section, we provide a proof of Theorem~\ref{thm:MIR}. Our proof technique is based on techniques developed in~\cite{DanielySS15}. In \cite{DanielySS15}, they present a generalization of VC dimension and used it to unify all previous inapproximability results for maximal-in-range mechanisms. 

\textbf{Proof overview:} We will first show that if the first-round communication $C \leq m^{\delta}$, then the auctioneer is choosing a truthful mechanism from exponentially many but not too many maximal-in-range mechanisms. And then we generalize the results in~\cite{DanielySS15} to show that if the mechanism is potentially running a bounded number of maximal-in-range mechanisms and this mechanism has good approximation ratio, then one of the maximal-in-range mechanism it uses will shatter a large set. Finally, we will use results from~\cite{DanielySS15} to show that shattering a large set implies inefficiency in computation or communication by reducing from hard problems. 

The results of Daniely et. al. are not such that we can apply them in a black-box manner, and we must actually introduce and mildly generalize their techniques. We do so below.

\subsection{Notations}
\begin{definition}[\cite{DanielySS15}]
(Duplicate Allocation) A $d$-duplicate allocation of a set $X$ to indices $Y$ is a collection $\{S_y\}_{y \in Y}$ of subsets of $X$ such that every $x\in X$ belongs to at most $d$ subsets from $\{S_y\}_{y\in Y}$. We use $P_d(X,Y)$ to denote the set of all $d$-duplicate allocations. If $d=1$, we will use $P(X,Y)$. 
\end{definition}

\begin{definition}
(Maximal-In-Range Mechanism) Given a allocation set $\mathcal{H} \subseteq P_d(X,Y)$, A maximal-in-range (MIR or VCG-based) mechanism outputs an allocation $S$ that maximizes $\sum_{i=1}^n v_i(S_i)$ over all allocations in $\mathcal{H}$. The mechanism will charge the VCG price (induced by $\mathcal{H}$) to each bidder to ensure truthfulness.
\end{definition}

\begin{definition}[\cite{DanielySS15}]
(Shattering) Let $\mathcal{H} \subseteq P_d(X,Y)$. A pair of subsets $S \subseteq X$, $A \subseteq Y$ is shattered by $\mathcal{H}$, if $A^S = \mathcal{H}_{S,A}$. Here $\mathcal{H} _{S,A} :=\{f:S\rightarrow A |\exists \{S_y\}_{y \in Y} \in \mathcal{H} ~\forall y\in A,f^{-1}(y) = S_y \cap S\}$. That is, $\mathcal{H}_{S, A}$ is the set of allocations of items in $S$ to players in $A$ that are projections of allocations in $\mathcal{H}$ onto items in $S$ and players in $A$. $S$ and $A$ are shattered by $\mathcal{H}$ if $\mathcal{H}$ induces every possible allocation of $A$ to $S$.
\end{definition}

\begin{definition}[\cite{DanielySS15}]
(Generalization of VC-dimension) Let $\mathcal{H} \subseteq Y^X$. $A$ is $k$-shattered by $\mathcal{H}$ if $\forall a\in A$, $\exists Y_a \subseteq Y$ of size $k$ and the following holds: For each $a\in A$, choose $y_a \in Y_A$ arbitrarily, then no matter how you choose these $y_a$'s, we can find $f\in \mathcal{H}$ such that $\forall a \in A, f(a) = y_a$. We define $\Dim_k(\mathcal{H})$ as the maximal cardinality of a $k$-shattered set. 
\end{definition}

\begin{definition}[\cite{DanielySS15}]
(Approximate Containment Let $\mathcal{H}\subseteq P_d(X,Y)$ and $\alpha \geq 1$. $\mathcal{H}$ has the $\alpha$-containment property if $\forall$ allocations $\{S_y\}_{y\in Y}$ of $X$, $\exists \{ T_y\}_{y\in Y} \in \mathcal{H}$ such that $|\{y | \emptyset \neq S_y \subseteq T_y\}| \geq \frac{1}{\alpha} |\{y|\emptyset \neq S_y\}|$. In other words, $\mathcal{H}$ has the $\alpha$-containment property if for any possible allocation, $S$, there is an allocation $T$ in $\mathcal{H}$ such that rat least a $1/\alpha$ fraction of players who get something in $S$ get in $T$ at least what they got in $S$.
\end{definition}

\begin{definition}[\cite{DanielySS15}]
(Approximate Intersection)  Let $\mathcal{H}\subseteq P_d(X,Y)$ and $\alpha \geq 1$. $\mathcal{H}$ has the $\alpha$-intersection property if $\forall$ allocations $\{S_y\}_{y\in Y}$ of $X$, $\exists \{ T_y\}_{y\in Y} \in \mathcal{H}$ such that $\sum_{y \in Y} |S_y \cap T_y| \geq \frac{1}{\alpha}\sum_{y \in Y} |S_y|$.
\end{definition}

\subsection{Proof}
We first prove the following simple lemma to show that if the first-round communication $C \leq m^{\delta}$ per bidder, the protocol is choosing a truthful mechanism from exponentially many but not too many maximal-in-range mechanisms. It is easy to get the following observations:
%\begin{observation}
%\label{ob:norand}
%In order to make the mechanism achieve some approximation ratio on all inputs, there is no need for Merlin to use randomness.
%\end{observation}
%\begin{proof}
%\end{proof}
\begin{observation}
\label{ob:num}
If a protocol has first-round communication $ m^{\delta}$ per bidder, then the number of possible truthful mechanisms that can be selected for phase two is at most $2^{n\cdot m^{\delta}}$.
\end{observation}

It is easy to get the following observation from the definition of MIR mechanisms:
\begin{observation}
\label{ob:mir}
If phase two is choosing from $t$ MIR mechanisms with sets of allocations $\mathcal{H}^1,...,\mathcal{H}^t$, then the resulting no matter how phase one selects an MIR mechanism, the combined interpolation mechanism does no better than the MIR mechanism with set of allocations $\mathcal{H} = \cup_{i=1}^t \mathcal{H}^i$. 
\end{observation}
By this observation, we will start to consider the MIR mechanism with set of allocations $\mathcal{H}$. We are going to show that if the MIR mechanism with set of allocations $\mathcal{H}$ has good approximation ratio, then one of the maximal-in-range mechanism among $\mathcal{H}^1,...,\mathcal{H}^t$ will shatter a large set. We need the generalization of the Sauer-Shelah lemma from \cite{DanielySS15}:
\begin{theorem}[\cite{DanielySS15}]
 For every $\mathcal{H} \subseteq Y^X$, and every $2\leq k\leq |Y|$,
$$
|\mathcal{H}| \leq \sum_{i=0}^{Dim_k(\mathcal{H})}\binom{|X|}{i}(k-1)^{|X|-i}\binom{|Y|}{k}^i$$
$$ \leq |X|^{\Dim_k(\mathcal{H})}|Y|^{k\Dim_k(\mathcal{H})}(k-1)^{|X|}.$$
\end{theorem}

We will generalize Theorem 2.6 and Theorem 2.8 in \cite{DanielySS15} as the following two theorems.
\begin{theorem}
\label{thm:containment}
Let $\epsilon > 0$, $\mathcal{H}^1,...,\mathcal{H}^t \subset P_d(X,Y)$ and $\mathcal{H} = \cup_{i=1}^t \mathcal{H}^i$. Assume that $n = m^{1 -3\epsilon/4}$ and $t =e^{\frac{m^{3\epsilon/4}}{4}}$. If $\mathcal{H}$ has the $m^{1-\epsilon}$-containment property, then there exists $i \in [t]$ and a shattered pair $S\subseteq X$, $A \subseteq Y$ of sizes $\tilde{\Omega}(m^{3\epsilon/4})$ and $m^{\epsilon/4}$, such that the pair $S$,$A$ is shattered by $\mathcal{H}^i_{S,A}$.
\end{theorem}
\begin{proof}
Let $k = m^{\epsilon/4}$. Let $\{T_y\}_{y\in Y}$ to be  a partition that is uniformly randomly chosen from all the partitions that give $\frac{m}{n} = m^{3\epsilon/4}$ items to each index. Let $A$ to be chosen from all the sets of indices of size $k$ uniformly randomly and independently from $\{T_y\}_{y \in Y}$.

Since $\mathcal{H}$ has the $m^{1-\epsilon}$-containment property, there is a partition $\{S_y\}_{y\in Y} \in \mathcal{H}$ such that $T_y \subseteq S_y$ for at least $k$ $y$'s. We call these $y$'s are covered. The probability that $T_y \subseteq S_y$ for all $y \in A$ is $\geq \frac{1}{\binom{n}{k}} \geq \frac{1}{\binom{m}{k}} \geq m^{-k}$. By linearity of expectation, we can find a fixed set $A$ of $k$ indices such that with probability  $\geq m^{-k}$ (over the randomness in $\{T_y\}_{y \in Y}$), $A$ is covered. We fix this $A$. 

By linearity of expectation again, there exists a set $S \subset X$ and $|S| = m^{\epsilon}$ such that if we only count $A$ as covered when $\cup_{y \in A}T_y = S$, the probability that $A$ is covered is still $\geq m^{-k}$. Fix this $S$. If a function $f:S \rightarrow A$ is uniformly randomly chosen from all functions with the property that $\forall i,j\in A$, $|f^{-1}(i)| = |f^{-1}(j)|$, then with probability $\geq m^{-k}$, $f\in \mathcal{H}_{S,A}$.  The number of these functions is $k^{m^{\epsilon}}\cdot m^{-\epsilon k}$, $m^{-\epsilon k}$ comes from the fact that a random function from $S$ to $A$ will satisfy $\forall i,j\in A$, $|f^{-1}(i)| = |f^{-1}(j)|$ with probability $m^{-\epsilon k}$. 

Then we have $|\mathcal{H}_{S,A}| \geq k^{m^{\epsilon}} \cdot m^{-k} \cdot m^{-\epsilon k}$. By an averaging argument, there exists $i \in [t]$, such that $|\mathcal{H}^i_{S,A}| \geq \frac{1}{t} \cdot k^{m^{\epsilon}} \cdot m^{-k} \cdot m^{-\epsilon k}$. Fix this $i$.

By the generalized Sauer-Shelah Lemma above, we have (note that we are using just the first inequality in the Theorem statement, and that $|Y| = k$)
\[
|\mathcal{H}^i_{S,A}|   \leq (k-1)^{m^{\epsilon}} \cdot m^{\epsilon  \Dim_k(\mathcal{H}_{S,A}^i)}.
\]
Then we have
\begin{eqnarray*}
m^{\epsilon   \Dim_k(\mathcal{H}_{S,A}^i)} &\geq& \frac{1}{t} \cdot (\frac{k}{k-1})^{m^{\epsilon}}\cdot m^{-k}\cdot m^{-\epsilon k}\\
&\geq& e^{\ln(1+\frac{1}{k-1})\cdot m^{\epsilon} - \ln t - 2k\ln(m)}\\
&\geq& e^{\frac{1}{2k}\cdot m^{\epsilon} - \frac{m^{3\epsilon/4}}{4} - 2k\ln(m)}\\
&=& e^{\frac{m^{3\epsilon/4}}{4} - 2m^{\epsilon/4}\ln(m)}.\\
\end{eqnarray*}
By taking logarithms on both side, we have
\[
\Dim_k(\mathcal{H}^i_{S,A} ) = \tilde{\Omega}(m^{3\epsilon/4}).
\]
\end{proof}

\begin{theorem}
\label{thm:intersection}
Let $k \geq 2$ be some integer, $0 < \epsilon < \frac{1}{k+1}$, $\mathcal{H}^1,...,\mathcal{H}^t \subset P(X,Y)$ and $\mathcal{H} = \cup_{i=1}^t \mathcal{H}^i$. Assume $n = m^{\frac{1}{k+1} -\epsilon}$ and $t = (1+k)^{\frac{m}{n}}$. If $\mathcal{H}$ has the $\frac{m^{\frac{1}{k+1} -\epsilon}}{1+2k}$-intersection property, then there exists $i \in [t]$ and a shattered pair $S\subseteq X$, $A \subseteq Y$ of sizes $\frac{m^{(k+1)\epsilon}}{\log_2(m)}$ and $k$, such that the pair $S$,$A$ is shattered by $\mathcal{H}^i_{S,A}$.
\end{theorem}

\begin{proof} We use the following lemma from \cite{BuchfuhrerDFKMPSSU10}

\begin{lemma}
Assume $\mathcal{H} \subseteq P(X,Y)$ has the $(\frac{n}{1+2k})$-intersection property. There exists a set $S \subseteq X$ such that $|S| \geq km/n$ and $|\mathcal{H}_{S,Y}| \geq (1+k)^{km/n}$.
\end{lemma}
By this lemma, there is a set $S' \subseteq X$ such that $|S'| \geq km/n$ and $|\mathcal{H}_{S',Y}| \geq (k+1)^{km/n}$. By averaging argument, there exists $i \in [t]$ such that 
$$
|\mathcal{H}^i_{S',Y}| \geq \frac{1}{t} \cdot (k+1)^{km/n} = (k+1)^{-m/n} \cdot (k+1)^{km/n}$$ $$ \geq (k+1)^{m/n}.
$$
By generalized Sauer-Shelah Lemma, we have 
$$
(k+1)^{m/n} \leq m^{\Dim_k(\mathcal{H}^i_{S',Y})}\cdot  n^{k\Dim_k(\mathcal{H}^i_{S',Y})}\cdot (k-1)^{m/n} $$ $$\leq m^{2\Dim_k(\mathcal{H}^i_{S',Y})}\cdot (k-1)^{m/n}.
$$
Taking logarithms we have
\[
\Dim_k(\mathcal{H}^i_{S',Y}) \geq \frac{\log_2(\frac{k+1}{k-1})\frac{m}{n}}{2\log_2(m)} \geq \frac{m}{2n(k-1)\log_2(m)}.
\]
Therefore, there exists a set $S'' \subseteq S'$ of size $\frac{m}{2n(k-1)\log_2(m)}$ that is $k$-shattered by $Y$. By definition of $k$-shattering, let $\{Y_a\}_{a\in T}$(each $Y_a$ has size $k$) be a collection of subsets of $Y$ that indicates $S''$ is $k$-shattered. By averaging argument, since the total number of different $Y_a$ is at most $\binom{n}{k} \leq \frac{n^k}{(k-1)^2}$, there is a subset $S \subseteq S''$ of size $\frac{\frac{m}{n^{k+1}}}{\log_2(m)} \geq \frac{\frac{m}{m^{1-(k+1)\epsilon}}}{\log_2(m)} = \frac{m^{(k+1)\epsilon}}{\log_2(m)}$ such that for all $a\in S$, $Y_a$ are the same. Let this $Y_a$ to be $A$, we have that the pair $S$,$A$ is shattered by $\mathcal{H}^i_{S,A}$.

\end{proof}

\begin{observation}[\cite{DanielySS15}]
\label{ob:conint}
$\mathcal{H}$ has the $\alpha$ containment property if and only if the MIR mechanism with set $\mathcal{H}$ has an approximation ratio of $\alpha$ with respect to single minded valuations. $\mathcal{H}$ has the $\alpha$ intersection property if and only if the MIR mechanism with set $\mathcal{H}$ has an approximation ratio of $\alpha$ with respect to 0/1-additive valuations.\footnote{A valuation is $0/1$-additive if it is additive with all item values either $0$ or $1$.}
\end{observation}

By putting Observation \ref{ob:num}, Observation \ref{ob:mir}, Observation \ref{ob:conint}, Theorem \ref{thm:containment} and Theorem \ref{thm:intersection} together, we get the following two Corollaries:

\begin{corollary}
For all $\delta > 0$, if an MIR interpolation mechanism has first round communication at most $m^{\delta}$ per bidder and guarantees  approximation ratio $m^{1/3 - 2\delta/3}$ when all buyers are single minded, then one of the MIR mechanisms selected for phase two maximizes over a set of allocations $\mathcal{H}^i$ that shatters  a pair of subsets $S \subseteq [m]$ and $A \subseteq [n]$ of sizes $\tilde{\Omega}(m^{1/2 + \delta/2})$ and $m^{1/6+\delta/6}$, respectively. 
\end{corollary}

\begin{corollary}
For all $\delta > 0$, if a MIR interpolation mechanism has first round communication at most  $m^{1/3}$ per bidder and guarantees approximation ratio $\frac{m^{1/3 -\delta}}{5}$ when all buyers are 0/1-additive, then one of the MIR mechanisms selected for phase two maximizes over a set of allocations $\mathcal{H}^i$ that shatters  a pair of subsets $S \subseteq [m]$ and $A \subseteq [n]$ of sizes $\frac{m^{3\delta}}{\log_2(m)}$ and $2$.
\end{corollary}

Finally, these two corollaries provide the proof of Theorem~\ref{thm:MIR}.

\begin{prevproof}{Theorem}{thm:MIR}
First, it is clear that both capped-additive and submodular valuations contain 0/1-additive valuations. At a high level, the fact that $\mathcal{H}^i$ shatters a large pair of subsets means that any mechanism that optimizes over $\mathcal{H}^i$ can solve instances over $S, A$ \emph{exactly}. If it is computationally hard (or impossible with low communication) to optimize over sets of the corresponding sizes exactly, then it must be the case that the MIR mechanism optimizing over $\mathcal{H}^i$ isn't computationally efficient (or communication efficient). 

To prove the three bullets formally, one just needs to encode problems that are NP-hard to solve (or require exponential communication) into welfare maximization for the appropriate valuation classes. Daniely et. al. provides such reductions for each of the three problems, so we use exactly the same reductions as them. We refer the reader to~\cite{DanielySS15} for further detail.
\end{prevproof}

\section{Omitted Proofs from Section~\ref{sec:value}}\label{app:general}

In this section, we consider the computationally efficient and value query interpolation mechanisms. In \cite{Dobzinski11},\cite{DobzinskiV12}, they showed that if such a truthful mechanism has a good approximation ratio with respect to submodular valuations, then its query complexity and computational complexity will be large. We will generalize their techniques to give lower bounds in our setting.

\textbf{Proof overview:}  Similarly as the proof in the previous section, we will first show that if the first round communication is small, then the number of truthful mechanisms that might be run in phase two is bounded. Then we will show that one of these truthful mechanism has a large structured submenu on a random polar additive valuation with constant probability. The structured submenu will be defined later. Finally we will use \cite{Dobzinski11} and \cite{DobzinskiV12}'s results to show that a large structured submenu implies large query complexity and computational complexity.

\subsection{Notations}
For any truthful mechanism $A$, we use $A(v)_i$ to denote the set $A$ gives to bidder $i$ on valuation $v$ and we use $p^A_{v_{-i}}(S)$ denote the price $A$ charges bidder $i$ for getting item set $S$ on valuation $v_{-i}$. Notice the reason that we can define the prices in this way is because of the``taxation principle" for truthful mechanisms: $p^A_{v_{-i}}(S)$ does not depend on bidder $i$'s own declared valuation. 
\begin{definition}[\cite{Dobzinski11}]
(structured submenu) A collection of sets $\mathcal{S}(A, v_{-i}, k,p)$ is called a structured submenu if 
\begin{enumerate}
\item $\mathcal{S}(A, v_{-i}, k,p) \subseteq \{S|\exists v_i, A(v_i,v_{-i})_i = S\}$.
\item For all $S \in \mathcal{S}(A, v_{-i}, k,p)$, $|S| = k$. 
\item For all $S \in \mathcal{S}(A, v_{-i}, k,p)$, $p-\frac{1}{m^5} < p^A_{v_{-i}}(S) \leq p$.
\item For all $T \in \{S|\exists v_i, A(v_i,v_{-i})_i = S\}$, $S \in \mathcal{S}(A, v_{-i}, k,p)$ such that $T$ strictly contains $S$, $p^A_{v_{-i}}(S) + \frac{1}{m^3}  \leq p^A_{v_{-i}}(T)$.
\end{enumerate}
\end{definition}

\begin{definition}[\cite{Dobzinski11}]
(polar additive valuations) A valuation $v$ is called polar additive if 
\begin{enumerate}
\item $v$ is additive. 
\item For each item $j$ either $v(\{j\}) = 1$ or $v(\{j\}) = 1/m^3$. 
\end{enumerate}
\end{definition}

\subsection{Proof}
\cite{Dobzinski11,DobzinskiV12} prove that any approximately optimal truthful mechanism has a large structured submenu. We first prove the following lemma to show that if a protocol with first-round communication achieves a good approximation, we can still find a large structured submenu. 
\begin{lemma}
\label{lem:largemenu}
Let $m \geq n^2$. Let $A$ be an interpolation mechanism with first round communication $x$ per bidder. that guarantees approximation ratio $\frac{n}{20}$ whenever buyers are submodular. Let $t = 2^{nx}$, and let the $t$ truthful mechanisms that might be selected by phase one be $A_1$, ...,$A_t$. Let $v$ be a polar additive valuation randomly drawn from a distribution described in the proof. Then with probability $1 - e^{-\frac{m}{1000}} - t\cdot \frac{e^{-\frac{2m}{n^2}}}{20n}$, there exists $i \in \{1,...,n\}$, $l \in \{1,...,t\}$, $k \in \{1,...,m\}$,  $p \in [0,m]$, with $p$ a multiple of $1/m^5$, and a structured submenu $\mathcal{S}(A_l, v_{-i}, k,p)$, such that $|\mathcal{S}(A_l, v_{-i}, k,p)| \geq \frac{e^{\frac{m}{n^2}}}{10n^2\cdot m^7}$.
\end{lemma}
\begin{proof}
We pick a random polar additive valuation $v$ as follows: for each item $j$ and bidder $i$, $v_i(\{j\}) = 1$ with probability $\frac{1}{n}$ and otherwise $v_i(\{j\}) = \frac{1}{m^3}$. 

For each truthful mechanism $A_l$, define 
\[
\mathcal{S}^{A_l}_{v_{-i}} = \{S|\exists \text{ polar additive } v_i, A_l(v_i,v_{-i})_i = S\}.
\]
Then for all $S \in \mathcal{S}^{A_l}_{v_{-i}}$, we have $p_{v_{-i}}(S) \leq m$. This is because otherwise bidder $i$ will prefer the empty bundle. For all $S \in \mathcal{S}^{A_l}_{v_{-i}}$, $T \in \{S|\exists v_i, A_l(v_i,v_{-i})_i = S\}$ and $T$ strictly contains $S$, then $p^{A_l}_{v_{-i}}(S) + \frac{1}{m^3}  \leq p^{A_l}_{v_{-i}}(T)$. This is because each item in a polar additive valuation has value at least $\frac{1}{m^3}$. Thus if the price difference between $S$ and $T$ is less than $\frac{1}{m^3}$ then bidder $i$ cannot be truthful. 

Now let's consider the following three events and prove they cannot happen at the same time:
\begin{enumerate}
\item $D \geq (1-1.1(1-\frac{1}{n})^n)m$ where $D = |\{j|\exists i, v_i(\{j\}) = 1\}|$.  
\item For all $i$ and $l$, the lexicographically small $\frac{e^{\frac{m}{n^2}}}{10n^2}$ bundles in $\mathcal{S}^{A_l}_{v_{-i}}$ satisfy $v_i(S) \leq \max\{\frac{4|S|}{n} + \frac{1}{m^2}, \frac{4m}{n^2} + \frac{1}{m^2}\}$. 
\item For all $i$ and $l$, $|\mathcal{S}^{A_l}_{v_{-i}}| \leq \frac{e^{\frac{m}{n^2}}}{10n^2}$.
\end{enumerate}

Now let's see why these events cannot happen at the same time. When the first event happens, we know that the optimal assignment will make the social welfare to be at least $m(1-\frac{1.1}{e})$. The second and third event together says that for all $l$,$i$ and $S \in \mathcal{S}^{A_l}_{v_{-i}}$, $v_i(S) \leq \max\{\frac{4|S|}{n} + \frac{1}{m^2}, \frac{4m}{n^2} + \frac{1}{m^2}\}$. Therefore, for each $A_l$, consider the output of $A_l$ as $A_l(v)_1,...,A_l(v)_t$. The social welfare of $A_l$ is 
\[
\sum_i v_i(A_l(v)_i) \leq\sum_i \max\{\frac{4|A_l(v)_i|}{n} + \frac{1}{m^2}, \frac{4m}{n^2} + \frac{1}{m^2}\}\]\[ \leq  (\frac{4m}{n^2} + \frac{1}{m^2}) \cdot n +\sum_{|A_l(v)_i| > n/m} (\frac{4|A_l(v)_i|}{n} + \frac{1}{m^2} ) \leq \frac{10m}{n}
\]
Therefore the approximation ratio of $A_l$ on this valuation $v$ is at least $(m(1-\frac{1.1}{e}) )/ (\frac{10m}{n}) \geq \frac{n}{20}$. Since this is true for all $A_l$, it contradicts with the fact that $A$ has approximation ratio $\frac{n}{20}$ with respect to submodular valuations. 

Now let's analyze the probability that both of the first two events to happen. For the first event, for each item $j$, the probability that $j$ is not in $\{j|\exists i, v_i(\{j\} = 1\}$ is $(1-\frac{1}{n})^n$. Therefore, by Chernoff Bound,
\[
Pr[D < (1-1.1(1-\frac{1}{n})^n)m] \leq e^{-\frac{0.1^2 (1-\frac{1}{n})^n m}{3}} < e^{-\frac{m}{1000}}.
\]
For second event, fix $S$ and $i$, let $D_{i,S} =  |\{j| j\in S, v_i(\{j\} = 1\}|$. By Chernoff Bound,
\[
Pr[D_{i,S} \geq \frac{4}{n} |S|]  \leq e^{-\frac{3^2\cdot |S|}{3n}} = e^{-\frac{3|S|}{n}}.
\]
It is easy to see that if $v_i(S) \geq \frac{4}{n}|S| + \frac{1}{m^2}$, then $D_{i,S} \geq \frac{4}{n}|S|$. If $|S| \geq m/n$, we have with probability at most $e^{-3m/n^2}$,  $v_i(S) \geq \frac{4|S|}{n} + \frac{1}{m^2}$. If $|S| < m/n$, then we extend $S$ to $T$ such that $|T| = m/n$ by adding items. We know that $Pr[D_{i,T} \geq \frac{4}{n} |T|] \leq e^{-3m/n^2}$ and $v_i(T) \geq v_i(S)$. Therefore with probability at most $e^{-3m/n^2}$,   $v_i(S) \geq \frac{4m}{n^2} + \frac{1}{m^2}$. So for any $S$, 
\[
Pr[v_i(S) > \max(\frac{4|S|}{n} + \frac{1}{m^2}, \frac{4m}{n^2} + \frac{1}{m^2})] < e^{-3m/n^2}.
\]
Thus by Union Bound, the probability that both of the first two events to happen is at least
\[
1- e^{-\frac{m}{1000}} - t \cdot n \cdot \frac{e^{\frac{m}{n^2}}}{10n^2} \cdot e^{-3m/n^2} = 1 - e^{-\frac{m}{1000}} - t\cdot \frac{e^{-\frac{2m}{n^2}}}{10n}.
\]
Therefore the third event happens with probability at most $e^{-\frac{m}{1000}} + t\cdot \frac{e^{-\frac{2m}{n^2}}}{10n}$. So with probability $1 - e^{-\frac{m}{1000}} - t\cdot \frac{e^{-\frac{2m}{n^2}}}{10n}$, there exists $i$ and $l$ such that $\mathcal{S}^{A_l}_{v_{-i}} \geq \frac{e^{\frac{m}{n^2}}}{10n^2}$. Now fix this $ \mathcal{S}^{A_l}_{v_{-i}}$. For each $S \in \mathcal{S}^{A_l}_{v_{-i}}$, we put it into bin $(k,p)$, if $|S| = k$ and $p-\frac{1}{m^5} < p^{A_l}_{v_{-i}}(S) \leq p$. There are only $m$ choices for $k$ and $m^6$ choices for $p$. So there are $m^7$ bins. Therefore, one bin contains at least $\frac{e^{\frac{m}{n^2}}}{10n^2m^7}$ elements. And from the definition of the structured submenu, we know that this bin is a structure submenu of size at least $\frac{e^{\frac{m}{n^2}}}{10n^2m^7}$. 
\end{proof}

Then we are going to use the following lemma from \cite{Dobzinski11} which shows that a large structured submenu implies large query complexity of the mechanism.
\begin{lemma}[\cite{Dobzinski11}]
\label{lem:querycom}
Let $A$ be a truthful mechanism for combinatorial auctions with submodular bidders. Let $\mathcal{S}$ be a structured submenu. Then, the number of value queries $A$ makes is at least $|\mathcal{S}| -1$. 
\end{lemma}

\begin{prevproof}{Theorem}{thm:valuequeries}
Similarly to observations in the previous section, phase one is choosing from at most $t=2^{n\times m^{\delta}}$ truthful mechanisms $A_1,...,A_t$. By Lemma \ref{lem:largemenu}, if we choose a random polar additive valuation, with probability $1 - e^{-\frac{m}{1000}} - t\cdot \frac{e^{-\frac{2m}{n^2}}}{20n} \geq 1 -e^{-\frac{m}{1000}} -\frac{1}{20n} > 0$ , we will get a large structured submenu for some truthful mechanism $A_l$. And by Lemma \ref{lem:querycom}, we know that $A_l$ has to make at least $\frac{e^{\frac{m}{n^2}}}{10n^2m^7}- 1 \geq \frac{e^{m^{1/3}}}{10m^8} - 1$ queries. This proves the first bullet. 

In \cite{DobzinskiV12}, their Theorem 2.1 shows that if we can get the large structured submenu of some truthful mechanism with constant probability on a random polar additive valuation, then the truthful mechanism is computationally inefficient unless $NP = RP$. Indeed, we have already showed that with constant probability we will get a large structured submenu. Using the same technique of Theorem 2.1 \cite{DobzinskiV12}, we can directly get the second bullet as well. We refer the reader to~\cite{DobzinskiV12} for further details.
\end{prevproof}

\end{document}